\documentclass[11pt]{article}

\usepackage{times}
\usepackage{amsthm}
\usepackage{array}
\usepackage{amssymb}
\usepackage{amsmath}
\usepackage{dsfont}
\usepackage{paralist}
\usepackage{graphicx}
\usepackage[]{units}
\usepackage{mathrsfs}
\usepackage{tikz} 
\usepackage[subrefformat=parens]{subcaption}

\newtheorem{definition}{Definition}
\newtheorem{proposition}{Proposition}

\newtheorem{theorem}{Theorem}

\newtheorem{example}{Example}
\allowdisplaybreaks
\newcommand{\argmax}{{{\mathrm{argmax}}}}
\newcommand{\indicator}{{{\mathds{1}}}}

\newcommand{\naturals}{{{\mathbb{N}}}}
\newcommand{\reals}{{{\mathbb{R}}}}
\newcommand{\prob}{{{\mathbb{P}}}}
\newcommand{\pow}{{{\mathcal{P}}}}
\newcommand{\mw}{{{\mathcal{R}_\mathrm{mult}}}}
\newcommand{\sw}{{{\mathcal{R}_\mathrm{dec}}}}

\newcommand{\ap}{{{\mathrm{ap}}}}

\newcommand{\acc}{{{\mathrm{A}}}}
\newcommand{\rej}{{{\mathrm{R}}}}
\newcommand{\owa}{{\textsc{OWA}}}

\newcommand{\important}{{{D_{\mathrm{im}}}}}
\newcommand{\np}{{\mathrm{NP}}}

\title{What Do We Elect Committees For? \\A Voting Committee Model for Multi-Winner Rules\footnote{The preliminary version of this paper was presented at the 24th International Joint Conference on Artificial Intelligence (IJCAI-2015) and on the 13th Meeting of the Society for Social Choice and Welfare in 2016 (SSCW-2016).}}

\author{Piotr Skowron \\ University of Oxford \\ Oxford, UK \\ piotr.skowron@cs.ox.ac.uk }
\date{}

\begin{document}

\maketitle

\begin{abstract}
We present a new model that describes the process of electing a group of representatives (e.g., a parliament) for a group of voters. In this model, called the voting committee model, the elected group of representatives runs a number of ballots to make final decisions regarding various issues. The satisfaction of voters comes from the final decisions made by the elected committee. Our results suggest that depending on a decision system used by the committee to make these final decisions, different multi-winner election rules are most suitable for electing the committee. Furthermore, we show that if we allow not only a committee, but also an election rule used to make final decisions, to depend on the voters' preferences, we can obtain an even better representation of the voters. 
\end{abstract}

\section{Introduction}

There are various scenarios where a \emph{group of representatives} is selected to make decisions on behalf of a larger population of voters. The examples of such situations include parliamentary elections, elections for supervisory or faculty board, elections for the trade union, etc. In such scenarios, the elected group of representatives, also referred to as a \emph{committee}, runs a sequence of ballots making decisions regarding various issues. For instance, these can be decisions made by the elected parliament regarding financial economics, national health-care system, retirement age, or changes in the specific law acts. In such case, it is natural to judge the quality of the elected committee based on the quality of its decisions. In this paper we explore this idea and introduce new approach that allows for a normative comparison of various multiwinner election rules.

We introduce and study a new formal model, hereinafter referred to as the \emph{voting committee model}. In this model we assume that an elected committee runs a sequence of independent ballots, in which it makes collective decisions regarding a given set of issues. We assume that the ultimate satisfaction of the voters depends solely on the final decisions made by the committee. Consequently, each voter ranks candidates for the committee based on how likely it is that they vote according to his or her preferences. The voting committee model allows to numerically assess qualities of committees, depending on what election rule (refered to as a \emph{decision rule}) these committees use to make the final decisions. Intuitively, the numerical quality of a committee $S$ estimates how likely it is that $S$ makes decisions consistent with voters' preferences. 

There are many works that study scenarios in which a group of representatives is elected to make certain decisions. To the best of our knowledge, however, starting with the famous Condorcet's Jury Theorem~\cite{citeulike:896755} most of these works~\cite{optimalVotingRules, RePEc:cla:levarc:1560, citeulike:3441219, Myerson97extendedpoisson} model scenarios where there exists some ground truth and when the decisions made by the elected committee are either objectively correct or wrong. Thus, these models refer to the process of selecting a group of experts. Our approach is different---there is no ground truth and the voters can differ with their opinions regarding various issues. Informally speaking, a committee is good if it well represent voters' subjective opinions rather than if it can solve certain specific problems having objectively good or bad solutions. 

\subsection*{Our Contribution}

In this work we introduce a new formal probabilistic model that relates voters' satisfactions from a committee to their satisfactions from the committee's decisions. In our model we define the notion of optimality of a multi-winner election rule, given that a certain decision rule is used by the elected committee to make the final decisions.
We consider two approaches in our probabilistic voting committee model. In the first approach we consider voters which can be correlated with respect to their preferences regarding various issues. We study two specific cases, namely:
\begin{inparaenum}[(i)]
\item when voters can be represented as uniformly distributed points on the interval (in such case points can represent, for instance, positions of the corresponding voters in the left-right political spectrum), and
\item when preferences of voters are taken from publicly available datasets containing distributions of votes in numerous real-life voting scenarios~\cite{conf/aldt/MatteiW13}.
\end{inparaenum} 

In this approach we use the voting committee model, and through computer simulations we obtain several interesting conclusions. In particular, we observe that representation-focused multi-winner election rules, such as the Chamberlin--Courant rule~\cite{ccElection}, make better decisions with respect to voters' preferences in comparison to the other committee selection rules that we study. Somehow surprisingly, according to our simulations the Chamberlin--Courant rule is even superior to Proportional Approval Voting, a multiwinner extension of the d'Hondt proportional method of apportionment. Further, we observe that in our experiments proportional committees make significantly better decisions than a single representative would make. Last, but not least, we argue that decisions made by the elected committee through majority rule are better than those made by the committee through the random dictatorship rule. 

We also consider the specific case when each voter has a single, primarily important for him or her, issue. In this case the probabilistic model collapses to the deterministic one. Here, we obtain theoretical justification for the following claims. We argue that the top-$K$ rules, i.e., scoring rules~\cite{you:j:scoring-functions} that select $K$ candidates with the highest total score assigned by voters, are suitable for electing committees that use the random dictatorship rule to make decisions. It might seem that the significance of this observation is compromised by the low applicability of randomized election rules. Such randomized decision-making processes, however, model situations where decisions are made by individual members of the committee but we are uncertain which issues will be considered by which individuals. We also observe that the median $\owa$ rule~\cite{owaWinner} is suitable for electing committees that make final decisions by majority voting. Informally speaking, in the median $\owa$ rule a voter is satisfied with a committee $S$ if he or she is satisfied with at least half of the members of $S$. We recall the formal definition of the class of the $\owa$ rules in Section~\ref{sec:preliminaries}. Further, our analysis suggests that representation focused multi-winner election rules  are particularly well-suited for electing committees that need to make unanimous decisions. Unanimity is often required in situations where making a wrong decision implies severe consequences, e.g., in case of juries voting on convictions. Indeed, in such cases we are particularly willing to ensure that minorities are represented in the committees to avoid biased decisions.

In the second approach in our probabilistic model we consider independent voters. In such case we prove the existence of the optimal voting rule and show how to construct one. Discussion of this approach leads to the new very interesting concept---to the notion of a \emph{full multiwinner rule}. Informally speaking, a full multiwinner rule is the concept that allows voters to elect not only committees, but also the rules used by these committees to make decisions. We discuss the existence of optimal full multiwinner rules in Section~\ref{sec:indirect_rules}.

\subsection*{Related Work}

Our research is most closely related to the work of Koriyama et~al.~\cite{RePEc:ucp:jpolec:doi:10.1086/670380}, who alike assume that the elected committee makes a number of decisions in a sequence of ballots. In such model, they compute the frequency, with which the will of each individual is implemented. The difference is, however, that Koriyama et~al. consider the party list setting and the problem of apportionment, i.e., of allocating seats to the parties according to the numbers of votes the parties receive from voters. The authors show that certain assumptions regarding how the society evaluates the frequencies with which the opinions of multiple individuals is implemented, justify degressive proportionality, an interesting principle of proportional apportionment.

Similarly to our work, Casella~\cite{journals/geb/Casella05, StorableVotesBook} considers the model where a committee meets regularly over time to make a sequence of binary decisions. She defines Storable Votes, a multiple-issue voting system aimed at promoting rights of minorities. The strategic behavior of committee members voting on various issues is also captured in an innovative model by Colonel Blotto games~\cite{Laslier2002106}. 

Our work extends the literature on properties of multi-winner election rules~\cite{Simeone.Pukelsheim2006MathematicsandDemocracy,barberaNonControversial,elk-fal-sko-sli:c:multiwinner-rules,justifiedRepresenattion,geh:j:condorcet-committee,deb:j:k-borda}. This literature includes, e.g., the works of Barber\'a and Coelho~\cite{barberaNonControversial}, where the authors define properties that ``good'' multi-winner rules should satisfy. Elkind~et~al.~\cite{elk-fal-sko-sli:c:multiwinner-rules} argue that the desirability of many natural properties of multi-winner rules must be evaluated in the context of their specific applications. Our paper extends this discussion by showing an intuitive model and concrete examples concerning this model, for which different multi-winner rules are particularly well applicable.

Similarly to Fishburn~\cite{fishburn1981}, we explore the idea of comparing multi-winner election systems. Further, our perspective is conceptually close to the ones given by Christian~et~al.~\cite{lobbying}, who study computational problems related to lobbying in direct democracy, where the decisions are made directly by the voters who express their preferences in open referenda (there are no representatives).

Many multi-winner election systems are defined as functions selecting such committees that optimize certain metrics, usually related to voters satisfaction~\cite{Bock1998219,ccElection,monroeElection,minimaxProcedure,pavVoting,deb:j:k-borda}. These different optimization metrics capture certain desired properties of election systems. Our paper complements these works by presenting specific metrics that are motivated by the analysis of decision-making processes of committees.

Finally, we note that there exists a broad literature on voting on multi-attribute domains, where agents need to make collective decisions regarding a number of issues~\cite{paradoxMultipleElections,RePEc:sae:jothpo:v:12:y:2000:i:1:p:5-31,Xia:2008:VMD:1619995.1620029,Xia10aggregatingpreferences}. We differ from these works by considering an indirect process of decision-making, through an elected committee.

\section{Illustrative Example}\label{sec:illustrativeExample}

Consider a multiwinner election with six voters, eight candidates, and where the goal is to select a committee of three representatives. Assume that for each voter there exists a single important issue that he or she cares about; an issue can be, for instance, the decision of decreasing the taxes, of decreasing the retirement age, or the decision of making a specific change in the national health care system. The elected committee will vote on each issue and make decision of either accepting or rejecting it; assume that the elected committee will make such decisions by majority voting. Further, assume that the voters know the views of the candidates on the issues and that a voter $i$ approves of a candidate $a$ if and only if $i$ and $a$ have the same view on the $i$'s important issue.

Consider example preferences of voters over candidates illustrated in Table~\ref{tab:exampleApprovalProfileA}. Observe that committee $S = \{c_1, c_2, c_3\}$ would make decisions consistent with preferences of voter $i_1$. Indeed, for the $i_1$'s important issue, the majority of $S$, namely candidates $c_1$ and $c_2$, have preferences consistent with $i_1$. Similarly, for the issues important for voters $i_2$, $i_3$, $i_4$ and $i_6$, committee $S$ would make decisions consistent with the preferences of these voters. Informally speaking, five out of six voters would be satisfied with the decisions made by committee $S$. Any other committee would make decisions satisfying less voters. Consequently, in this scenario committee $S$ should win the election.

\begin{table}[!tbh]
    \caption{Example of preferences of the voters over candidates. Values '0' and '1' in the fields denote that corresponding voters approve of or disapprove of the corresponding candidates, respectively. For instance, voter $i_1$ in Table~\subref{tab:exampleApprovalProfileA} approves of candidates $c_1$, $c_2$, $c_5$, and $c_8$. The same voter in Table~\subref{tab:exampleApprovalProfileB} approves of candidates $c_1$ and $c_8$.}
    \begin{subtable}{.48\linewidth}
      \centering
        \caption{}\label{tab:exampleApprovalProfileA}
        \begin{tabular}{c | c c c c c c c c} 
        \hline
        & ${\bf c_1}$ & ${\bf c_2}$ & ${\bf c_3}$ & $c_4$ & $c_5$ & $c_6$ & $c_7$ & $c_8$ \\
        \hline\hline
        $i_1$ & {\bf 1} & {\bf 1} & {\bf 0} & 0 & 1 & 0 & 0 & 1 \\ 
        \hline
        $i_2$ & {\bf 1} & {\bf 1} & {\bf 1} & 1 & 1 & 1 & 0 & 0 \\
        \hline
        $i_3$ & {\bf 1} & {\bf 0} & {\bf 1} & 1 & 0 & 0 & 1 & 0 \\
        \hline
        $i_4$ & {\bf 0} & {\bf 1} & {\bf 1} & 0 & 0 & 0 & 1 & 0 \\
        \hline
        $i_5$ & {\bf 1} & {\bf 0} & {\bf 0} & 1 & 1 & 1 & 0 & 0 \\
        \hline
        $i_6$ & {\bf 0} & {\bf 1} & {\bf 1} & 0 & 0 & 1 & 0 & 1 \\
        \hline
        \end{tabular}
    \end{subtable}%
    \quad
    \begin{subtable}{.48\linewidth}
      \centering
        \caption{}\label{tab:exampleApprovalProfileB}
        \begin{tabular}{c | c c c c c c c c} 
        \hline
        & ${\bf c_1}$ & ${\bf c_2}$ & $c_3$ & $c_4$ & $c_5$ & $c_6$ & $c_7$ & ${\bf c_8}$ \\
        \hline\hline
        $i_1$ & {\bf 1} & {\bf 0} & 0 & 0 & 0 & 0 & 0 & {\bf 1} \\ 
        \hline
        $i_2$ & {\bf 0} & {\bf 1} & 0 & 1 & 0 & 0 & 0 & {\bf 0} \\
        \hline
        $i_3$ & {\bf 1} & {\bf 0} & 0 & 0 & 0 & 0 & 0 & {\bf 0} \\
        \hline
        $i_4$ & {\bf 0} & {\bf 0} & 1 & 0 & 0 & 0 & 1 & {\bf 1} \\
        \hline
        $i_5$ & {\bf 1} & {\bf 0} & 0 & 0 & 1 & 0 & 0 & {\bf 0} \\
        \hline
        $i_6$ & {\bf 0} & {\bf 1} & 0 & 0 & 0 & 0 & 0 & {\bf 1} \\
        \hline
        \end{tabular}
    \end{subtable} 
\end{table}

The above arguments can be used to design a voting rule that would be optimal for dichotomous voters' preferences. In such a rule, a voter $i$ approves of a committee $S$ if majority of the members of $S$ are approved by $v$; the committee approved by most voters is announced as the winner. We will formalize this argument later in Theorem~\ref{thm:majority}.

Now, let us change our initial assumptions and consider the case when the elected committee uses random dictatorship rule to make decisions regarding issues. Let us recall that according to the random dictatorship rule, for each issue a single committee member is selected uniformly at random, and the selected member makes a decision on the issue. Observe that in such case, voter $i_1$ is less satisfied with committee $S = \{c_1, c_2, c_3\}$ than before. Indeed, on $i_1$'s important issue, committee $S$ will make decision consistent with $i_1$'s preferences only with probability equal to $\nicefrac{2}{3}$. The expected number of voters satisfied with the decisions made by committee $S$ is equal to four, and for any other committee the expected number of satisfied voters is even lower. We conclude that in this case it is socially more beneficial if the elected committee uses majority voting rather than random dictatorship rule. The natural question arises whether this is always the case. To answer this question consider preferences of voters from Table~\ref{tab:exampleApprovalProfileB}. For these preferences, the decisions made by the optimal committee via majority rule will satisfy only two out of six voters. On the other hand, committee $S' = \{c_1, c_2, c_8\}$ using random dictatorship would satisfy, in expectation, $2\nicefrac{2}{3}$ voters.

The above observations suggest that it is justified to consider committee selection rules where not only the winning committee, but also the rule used by the elected committee to make decisions, depends on the voters' preferences. In Section~\ref{sec:indirect_rules} we formalize this concept, by introducing \emph{full multiwinner rules}. A full multiwinner rule is a function that given voters' preferences returns a pair: a given-size committee and a randomized decision rule over the set of two alternatives (\emph{accept} and \emph{reject}).

In the further part of this paper we will explain how to generalize the above reasoning beyond dichotomous preferences of the voters. Briefly speaking, we will assume that the voters do not know how the candidates are going to vote over the issues. Instead, we will assume that each voter has some intuition about how well he or she can be represented by a given candidate; formally, we will assume that each voter $v$ is able to assess the probability with which a given candidate will vote according to $v$'s preferences. This captures, for instance, a common scenario where the issues are not known in advance, thus the voters must use their beliefs to decide on how well given candidates will represent them.

\section{The Voting Committee Model}

In this section we describe a voting committee model that formalizes intuitions given is Section~\ref{sec:illustrativeExample}, and that allows for a normative comparison of various multi-winner election rules. For each set $X$, by $\indicator_X$ we denote the indicator function of $X$, i.e., $\indicator_X(x) = 1$ if $x \in X$, and $\indicator_X(x) = 0$ if $x \notin X$. For simplicity, we write $\indicator_X$ instead of $\indicator_X(x)$ whenever $x$ is clear from the context.

Let $N = \{1, 2, \dots, n\}$ be a set of \emph{voters}, and let $C = \{c_1, c_2, \dots, c_m\}$ be a set of \emph{candidates}. We assume there is a set $D = \{D_1, D_2, \dots, D_r\}$ of $r$ issues; each issue $D_j$ is a binary set consisting of two alternatives $D_j = \{\acc, \rej\}$. Intuitively, the issue can be either accepted or rejected; $\acc$ and $\rej$ corresponds to accepting and rejecting it, respectively. Voters have strict preferences over the alternatives within each issue; by $d_i^j$ we denote the preferred alternative from $D_j$ from the perspective of voter $i$. The issues might differ in their importance to different voters. We consider two types of attitudes: a voter $i$ can consider an issue $D_j$ either as \emph{important} or as \emph{insignificant}. We denote the set of all the issues important for voter $i$ as $\important(i) \subseteq D$.

In the first stage in our model, a committee $S$ is selected through a multi-winner election rule $\mw$; the selected committee consists of $K$ members. In the second stage, the committee $S$ runs $r$ independent ballots, for each ballot using the same (randomized) decision rule $\sw$. In the $i$-th ballot, $1 \leq i \leq r$, the committee makes a collective decision regarding issue $D_i$---for $D_i$ the committee makes either decision $\acc$ or decision $\rej$. The final outcome of this two-stage process is described by a vector of $r$ decisions. Intuitively, the first stage of elections might correspond to e.g., parliamentary elections, elections for supervisory or faculty board, etc. An election in the second stage might be viewed as e.g., a parliamentary ballot on an issue regarding, e.g., financial and monetary economics, education politics, changes in national health-care system, etc.

A multiwinner election rule $\mw$ takes as input a matrix of scores that voters assign to the candidates and returns a subset of $K$ candidates.
The scores might be provided directly by the voters or extracted from their ordinal preferences through a positional scoring function (we will discuss this issue in more detail in the subsequent sections).
A \emph{(randomized) decision rule} $\sw$ is a function $\sw: \naturals \to \reals$, that for each natural number $a$, corresponding to the number of votes casted on $\acc$ by the committee members, returns the probability that value $\acc$ is selected. We require that $\sw$ is symmetric with respect to decisions $\acc$ and $\rej$, i.e., that for each $a$, $0 \leq a \leq K$, it holds that $\sw(a) = 1 - \sw(K-a)$.

The ultimate satisfaction\footnote{We write ``ultimate satisfaction'' instead of ``satisfaction'' to distinguish these values, representing utilities that voters get from decisions of a committee from the scores that voters assign to individual candidates, and which quantify the level of appreciation of voters to individuals.} of voter $i$ depends solely on the final outcome of the $r$ ballots. Intuitively, voter $i$ considers committee $S$ as good, if for $i$'s important issues, i.e., issues from $\important(i)$, $S$ is likely to make decisions consistent with $i$'s preferences.
The ultimate satisfaction of a voter $i$ from committee $S$ is measured by value $\prob_{S, \sw}(i)$, the probability that $S$ makes a decision consistent with $i$'s preferences assuming that $S$ uses rule $\sw$ to make final decisions. Throughout the paper we will consider several specific variants of the voting committee model which will differ in a way the value $\prob_{S, \sw}(i)$ is defined.

Let us now define the central notion of this paper that we will use in further analysis.
\begin{definition}\label{def:utilOptimal}
Let $\sw$ be a decision rule used by the committee to make final decisions, and let $K$ denote the size of the committee to be elected. A committee $S$ is \emph{optimal} if:
\begin{align}
S \in \argmax_{S' \subseteq C: |S'| = K} \sum_i\prob_{S', \sw}(i) \textrm{.}
\end{align}
\end{definition}

In the above definition we take the utilitarian approach---the optimal rule aims at maximizing the sum of the ultimate satisfaction of the voters.
Analogously, we can define committees \emph{optimal in the egalitarian sense}, by replacing ``sum'' with ``min'' in Definition~\ref{def:utilOptimal}, optimizing the ultimate satisfaction of the least satisfied voter. For the sake of concreteness, in this paper we focus on the utilitarian case only. 
Definition~\ref{def:utilOptimal} implicitly introduces the way of comparing different committees, explored in this paper. A committee $S$ is preferred over a committee $S'$ if $\sum_i\prob_{S, \sw}(i) > \sum_i\prob_{S', \sw}(i)$.

\section{Overview of Election Rules}\label{sec:preliminaries}
In this section we recall definitions of several known decision and multi-winner election rules that we use in our further analysis.

\subsection{Decision Rules}

We recall that a (randomized) decision rule $\sw$ is a function $\sw: \naturals \to \reals$, that for each natural number $a$, corresponding to the number of votes casted on $\acc$ (accept), returns the probability that the value $\acc$ will be selected. We recall that we require $\sw$ to be symmetric, that is for each $a$, $0 \leq a \leq K$, it must hold that $\sw(a) = 1 - \sw(K-a)$.
Below we recall the definition of popular decision rules that we will use in our analysis.

The \emph{uniformly random dictatorship} rule selects an alternative $d$ with probability proportional to the number of committee members who vote for $d$.
The \emph{majority} rule deterministically selects $\acc$ if at least half of the committee members vote for $\acc$ (in order to avoid issues related to tie-breaking, we will always use the majority rule for the odd number of votes). Additionally, we consider one rule which is not symmetric with respect to decisions $\acc$ and $\rej$---the \emph{unanimity} rule returns deterministically $\acc$ if and only if all committee members vote for $\acc$.

\subsection{Multiwinner Election Rules}

For the description of multiwinner rules we assume that voters express their preferences over candidates by providing scores: for a voter $i$ and candidate $c$, by $u_{i, c}$ we denote the score that $i$ assigns to $c$. Intuitively, $u_{i, c}$ quantifies the level of appreciation of voter $i$ for candidate $c$.
In the description of our results in the subsequent sections we will discuss where these scores come from, and their specific structure. There are two particularly interesting types of scores:
\begin{description}
\item[Approval scores.] We say that scores are approvals if for each voter $i$ and each candidate $c$, we have $u_{i, c} \in \{0, 1\}$. We call the setting with approval scores the \emph{approval model}. In the approval model, we say that a voter $i$ approves of a candidate $c$ if $u_{i, c}=1$. Otherwise we say that $i$ disapproves of $c$.
\item[Borda scores.] In the model with Borda scores, for each voter $i$ and each candidate $c$, we have $u_{i, c} \in \{0, \ldots, m-1\}$, and $u_{i, c} \neq u_{i, c'}$ for each $c \neq c'$. In other words, each voter assigns his or her most preferred candidate score of $m-1$, to his or her second most preferred candidate score of $m-2$, etc. One way in which the Borda scores can be elicitated, is by asking voters for their rankings of candidates, and by applying Borda positional scoring functions to such ordinal preferences (cf.~the work of Young~\cite{you:j:scoring-functions}).
\end{description} 
A \emph{score profile} is a matrix of the scores of all voters over all candidates. We denote the set of all score profiles as $\mathscr{U}$.
For each $j \in \naturals$, by $\pow_j(S)$ we denote the set of all subsets of $S$ of size $j$. We refer to the elements of $\pow_K(C)$ as to \emph{committees}. A \emph{multi-winner election rule} is a function $\mw: \mathscr{U} \to \pow_K(C)$ that for a given score profile of voters returns a committee of size $K$.

A significant part of results provided in this paper concerns $\owa$ rules (OWA stands for an ordered weighted average operator). OWA rules in the context of approval model were first mentioned in the 19th century in the early works of Danish astronomer and mathematician Thorvald~N.~Thiele~\cite{Thie95a} and than generalized to arbitrary score profiles by Skowron~et~al.~\cite{owaWinner}.
Below, we recall the definition of this remarkably general class of rules and describe several concrete examples of OWA rules.

For each voter $i$, each committee $S$, and each number $j$, $1 \leq j \leq K$, let $u_{i, S}(j)$ denote the score of the $j$-th most preferred candidate from $S$, according to $i$. In other words: $\{u_{i, S}(1), \dots, u_{i, S}(K)\} = \{u_{i, c}: c \in S\}$, and $u_{i, S}(1) \geq \dots \geq u_{i, S}(K)$. For instance, $u_{i, S}(1)$ is the score that $i$ assigns to his or her most preferred candidate in $S$, $u_{i, S}(2)$ is the score that $i$ assigns to his or her second most preferred candidate in $S$, etc.
For each voter $i$, each committee $S$, and each $K$-element vector $\alpha = \langle \alpha_1, \ldots, \alpha_K \rangle$, we define the $\alpha$-satisfaction of $i$ from $S$ as the ordered weighted average (OWA) of the scores of the members of $S$: 
\begin{align*}
\alpha(i, S) = \sum_{j = 1}^K \alpha_j u_{i, S}(j) \textrm{.}
\end{align*}
The $\alpha$-rule selects a committee $S$ that maximizes total $\alpha$-satisfaction of the voters $\sum_i \alpha(i, S)$.

Intuitively, the $\alpha$ vector provides a way of aggregating scores of the individual members of the committee to obtain the score of the committee as the whole. Indeed, many known multiwinner rules are in fact OWA rules:
\begin{description}
\item[Top-$K$ rule and $K$-Borda rule.] If $\alpha_{\mathrm{Top}K} = \langle 1, 1, \dots, 1\rangle$ we get a top-$K$ rule that selects $K$ candidates with the highest total scores. Such rules are referred to as the weakly separable rules~\cite{elk-fal-sko-sli:c:multiwinner-rules}. For Borda scores, the top-$K$ rule collapses to the well known $K$-Borda rule~\cite{deb:j:k-borda}.

\item[Chamberlin--Courant rule.] Chamberlin--Courant rule~\cite{ccElection} is an OWA rule defined by the weight vector $\alpha_{\mathrm{CC}} = \langle 1, 0, \ldots, 0\rangle$. Informally speaking, according to Chamberlin and Courant a voter cares only about his or her most preferred candidate in the committee; such a most preferred candidate is a representative of the voter in the elected committee. Initially, Chamberlin and Courant defined their rule for Borda scores. Using Chamberlin--Courant rule in the approval model was first suggested by Thiele~\cite{Thie95a}. Elkind~et~al.~\cite{elk-fal-sko-sli:c:multiwinner-rules} considered the Chamberlin--Courant rule for different types of score profiles, and refer to these types of rules as representation focussed rules~\cite{elk-fal-sko-sli:c:multiwinner-rules}.

\item[Proportional Approval Voting (PAV).]  PAV is defined as an OWA rule with the weight vector $\alpha_{\mathrm{PAV}} = \langle 1, \nicefrac{1}{2}, \ldots, \nicefrac{1}{K} \rangle$. This rule, developed by Thiele, has been used for a short period in Sweden during early 1900's. This specific sequence of harmonic weights ensures a certain level of proportionality of the rule~\cite{pavVoting, justifiedRepresenattion}; as a matter of fact, PAV is often considered as an extension of the d'Hondt method of apportionment~\cite{Puke14a} to the case when voters can vote for individual candidates rather than for political parties.

\item[$k$-median rule.] The $k$-median rule is the OWA rule defined by the vector $\alpha$ which has $1$ on the $k$-th positions and $0$ on the others. In other words, according to the $k$-median rule, the satisfaction of a voter from a committee $S$ is his or her satisfaction from the $k$-th most preferred member of $S$.
\end{description}

The class of OWA-based rules captures many other interesting election systems such as top-$K$-counting rules~\cite{fal-sko-sli-tal:c:top-k-counting}, separable rules~\cite{elk-fal-sko-sli:c:multiwinner-rules} etc. For more discussion on OWA based rules and their applications beyond voting systems we refer the reader to the work of Skowron~et~al.~\cite{owaWinner}.

\medskip \noindent
\textbf{Sequential OWA Rules.} Unfortunately, for many OWA-based rules, finding the winners is a computationally hard problem\footnote{Computational hardness of the Chamberlin--Courant rule was first
proved by Procaccia et al.~\cite{complexityProportionalRepr}. Betzler~et~al.~\cite{fullyProportionalRepr} showed that this problem is also hard from the perspective of parameterized complexity theory. For arguments referring to other weight vectors we refer the reader to to the work of Skowron~et~al.~\cite{owaWinner}}. Sequential OWA rules form an appealing, computationally easy, alternative for OWA rules. Let $\alpha$ denote the vector of $K$ weights. The sequential $\alpha$-rule proceeds as follows. It starts with an empty solution $S = \emptyset$, and in each of the $K$ consecutive steps it adds to $S$ the candidate that increases the $\alpha$-satisfaction of the voters most. Usually, sequential OWA rules provide a good way of approximating their optimal counterparts~\cite{owaWinner}, sometimes exhibiting even more interesting properties than the original OWA rules~\cite{elk-fal-sko-sli:c:multiwinner-rules}.

\section{Correlated Voters in the Voting Committee Model}\label{sec:correlatedVoters}

In this section we consider voters which are not independent with respect to their preferences over the set of issues.  
We start our analysis by considering voters and candidates which can be represented as points on the line of preferences. The concept of representing voters and candidates as points in the Euclidean space dates back to 1966~\cite{DavisHinich66,Plott1967}, and since than, due to its many natural interpretations, it received a considerable amount of attention in the social choice literature~\cite{enelow1984spatial,enelow1990advances,mckelvey1990,merrill1999unified,schofield2007spatial}.

Let us consider the following illustrative example. Consider the population where one third of voters are left-wing and these voters are represented on the line by point 0, one third are right-wing and represented by point 1, and one third are centrists and represented by point $\nicefrac{1}{2}$. There are three issues $\mathsf{L}$, $\mathsf{C}$, and $\mathsf{R}$: the left-wing voters would like $\mathsf{L}$ accepted, $\mathsf{R}$ rejected, and they do not care about $\mathsf{C}$; the right-wing voters would like $\mathsf{R}$ accepted, $\mathsf{L}$ rejected, and they consider issue $\mathsf{C}$ insignificant; the centrists voters only care about $\mathsf{C}$ being accepted. A left-wing candidate will vote for accepting $\mathsf{L}$, $\mathsf{C}$ and $\mathsf{R}$ with probabilities 1, $\nicefrac{1}{2}$, and 0, respectively. A right-wing candidate will vote for accepting these issues with probabilities 0, $\nicefrac{1}{2}$, and 1, respectively. A centrist candidate will vote for accepting them with probabilities $\nicefrac{1}{2}$, 1, and  $\nicefrac{1}{2}$, respectively. In particular, left-wing voters are perfectly represented by left-wing candidates and perfectly misrepresented by right-wing candidates. Consider the two committees: $S$ which consists of a left-wing, a centrist, and a right-wing candidate, and $Q$ that consists of three centrist candidates.

\begin{figure*}[!t!h]
\begin{minipage}[h]{0.48\linewidth}
  \centering
  \begin{tikzpicture}
    \draw (0,0) -- (6,0);
    \filldraw [gray] (0.0, 0) circle (4pt);
    \filldraw [gray] (3.0, 0) circle (4pt);
    \filldraw [gray] (6.0, 0) circle (4pt);
  \end{tikzpicture}
  \caption*{Proportional committee $S$}
\end{minipage}
\begin{minipage}[h]{0.48\linewidth}
  \centering
  \begin{tikzpicture}
    \draw (0,0) -- (6,0);
    \filldraw [gray] (2.8, 0) circle (4pt);
    \filldraw [gray] (3.0, 0) circle (4pt);
    \filldraw [gray] (3.2, 0) circle (4pt);
  \end{tikzpicture}
  \caption*{Centrist committee $Q$}
\end{minipage}
\end{figure*}

The probability that committee $S$ makes decision consistent with preferences of left-wing voters is equal to $\nicefrac{1}{2}$. The same probability for right-wing voters is also equal to $\nicefrac{1}{2}$, and for the centrists voters to $\nicefrac{3}{4}$ (the probability that at least one of the two committee members, the left-wing or the right-wing, will vote for accepting issue $\mathsf{C}$). The ultimate satisfaction of the voters from committee $S$ is thus equal to $\nicefrac{n}{3} \cdot (\nicefrac{1}{2} + \nicefrac{1}{2} + \nicefrac{3}{4})$. The centrist committee $Q$, however, will make decision consistent with left-wing, centrist and right-wing voters with probabilities equal to $\nicefrac{1}{2}$, 1, and $\nicefrac{1}{2}$, respectively. Thus, committee $Q$ results in a higher ultimate satisfaction of the voters. The surprising conclusion of the above example is that the centrist committee, in some situations, can represent the voters better than the proportional one. In particular, this example shows that it is not clear how the best decision-making committee should be aligned on the preference line. Below, we will explore this idea further, we will argue that the described phenomenon is specific to our example, and that in the considered cases the proportional committee is usually, in some sense, superior.  

\subsection{Uniform Distribution on the Euclidean Line}\label{sec:prefOnLine}

Let us first consider the case where voters and candidates are uniformly distributed on the $[0, 1]$ interval. For this distribution with $n = 500$ voters and $m = 500$ candidates,
we run the following computer simulations:
\begin{enumerate}
\item For each voter $v$ we assumed that there exists one issue which can be represented on the line by the same point as the corresponding voter.
\item For each issue $i$ we uniformly at random selected a value $p \in [\nicefrac{3}{2}, \nicefrac{5}{2}]$. Intuitively, values of $p$ closer to $\nicefrac{3}{2}$ make the issue more likely to be preferred by  a majority of voters, thus such an issue can be considered as objectively desired. Values of $p$ closer to $\nicefrac{5}{2}$ suggest that the issue is likely to be preferred only by a minority of voters. Specifically, for each voter and for each candidate we randomly selected their preference over the issue. For an issue $i$ and a voter (or a candidate) $x$ we took $1 - p \cdot |i - x|$ as the probability of $x$ accepting $i$, where $|i-x|$ is the distance between points corresponding to $i$ and $x$ on the line. In particular, each voter always accepts his or her corresponding issue; further, for $p = 2$ and a centrist issue $i$, in expectation half of the voters (and half of the candidates) accept $i$.
\item For each issue we computed decisions made by different committees with respect to such issues, and we assigned to each voter $v$ the satisfaction equal to the fraction of issues for each a given committee made decisions consistent with $v$'s preferences.
\item We repeated such experiment 500 times and for each voter we computed the average ultimate satisfaction.  
\end{enumerate}
In these simulations we considered five different committees. Three of the considered committees consisted of 51 candidates and were selected by the top-$K$ rule, and by the sequential variants of PAV, and Chamberlin--Courant rules. To compute committees according to these rules, we assumed that the scores that voters assign to candidates are proportional to one minus the distance between the respective points. For a better intuition, below we depict example committees returned by the three considered multiwinner election rules (for the sake of readability of the diagram, below we depict the case of 10 rather than 51 committee members).

\begin{figure*}[!t!h]
\begin{minipage}[h]{0.33\linewidth}
  \centering
  \begin{tikzpicture}
    \draw (0,0) -- (3.5,0);
    \filldraw [gray] (0.1547, 0) circle (2pt);
    \filldraw [gray] (0.5663, 0) circle (2pt);
    \filldraw [gray] (0.8848, 0) circle (2pt);
    \filldraw [gray] (1.1851, 0) circle (2pt);
    \filldraw [gray] (1.5477, 0) circle (2pt);
    \filldraw [gray] (1.8487, 0) circle (2pt);
    \filldraw [gray] (2.1938, 0) circle (2pt);
    \filldraw [gray] (2.4864, 0) circle (2pt);
    \filldraw [gray] (2.8441, 0) circle (2pt);
    \filldraw [gray] (3.2543, 0) circle (2pt);
  \end{tikzpicture}
  \caption*{Chamberlin--Courant}
\end{minipage}
\begin{minipage}[h]{0.33\linewidth}
  \centering
  \begin{tikzpicture}
    \draw (0,0) -- (3.5,0);
    \filldraw [gray] (0.5901, 0) circle (2pt);
    \filldraw [gray] (0.8505, 0) circle (2pt);
    \filldraw [gray] (1.1501, 0) circle (2pt);
    \filldraw [gray] (1.3874, 0) circle (2pt);
    \filldraw [gray] (1.6282, 0) circle (2pt);
    \filldraw [gray] (1.8466, 0) circle (2pt);
    \filldraw [gray] (2.1049, 0) circle (2pt);
    \filldraw [gray] (2.3506, 0) circle (2pt);
    \filldraw [gray] (2.618, 0) circle (2pt);
    \filldraw [gray] (2.8413, 0) circle (2pt);
  \end{tikzpicture}
  \caption*{PAV}
\end{minipage}
\begin{minipage}[h]{0.33\linewidth}
  \centering
  \begin{tikzpicture}
    \draw (0,0) -- (3.5,0);
    \filldraw [gray] (1.3734, 0) circle (2pt);
    \filldraw [gray] (1.4434, 0) circle (2pt);
    \filldraw [gray] (1.5134, 0) circle (2pt);
    \filldraw [gray] (1.5834, 0) circle (2pt);
    \filldraw [gray] (1.6534, 0) circle (2pt);
    \filldraw [gray] (1.7234, 0) circle (2pt);
    \filldraw [gray] (1.7934, 0) circle (2pt);
    \filldraw [gray] (1.8634, 0) circle (2pt);
    \filldraw [gray] (1.9334, 0) circle (2pt);
    \filldraw [gray] (2.0034, 0) circle (2pt);
  \end{tikzpicture}
  \caption*{Top-$K$}
\end{minipage}
\end{figure*}
\noindent
Further, we considered one committee that consisted of a single centrist candidate, and one committee that consisted of all voters, i.e., the committee which represents the direct democracy---the case where all decisions are made in referenda.
Finally, we considered two different decision rules that the elected committees used to make final decisions with respect to issues (cf. point 3 above), namely the majority rule and the random dictatorship rule. 

\begin{figure*}[!t!h]
\begin{minipage}[h]{0.48\linewidth}
  \centering
  \includegraphics[width=\textwidth]{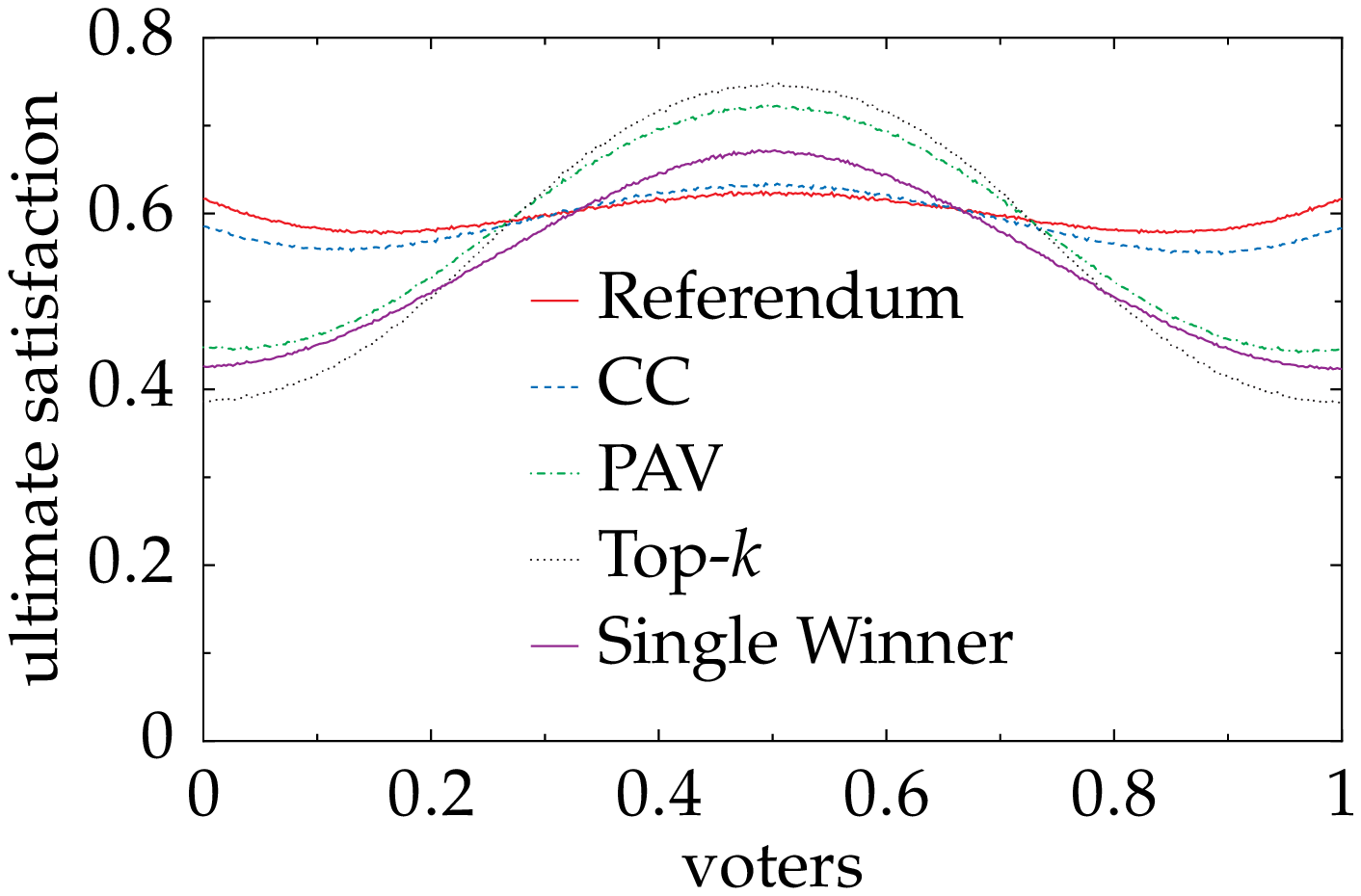}
  \caption*{(a) majority}
\end{minipage}
\hspace{0.3cm}
\begin{minipage}[h]{0.48\linewidth}
  \centering
  \includegraphics[width=\textwidth]{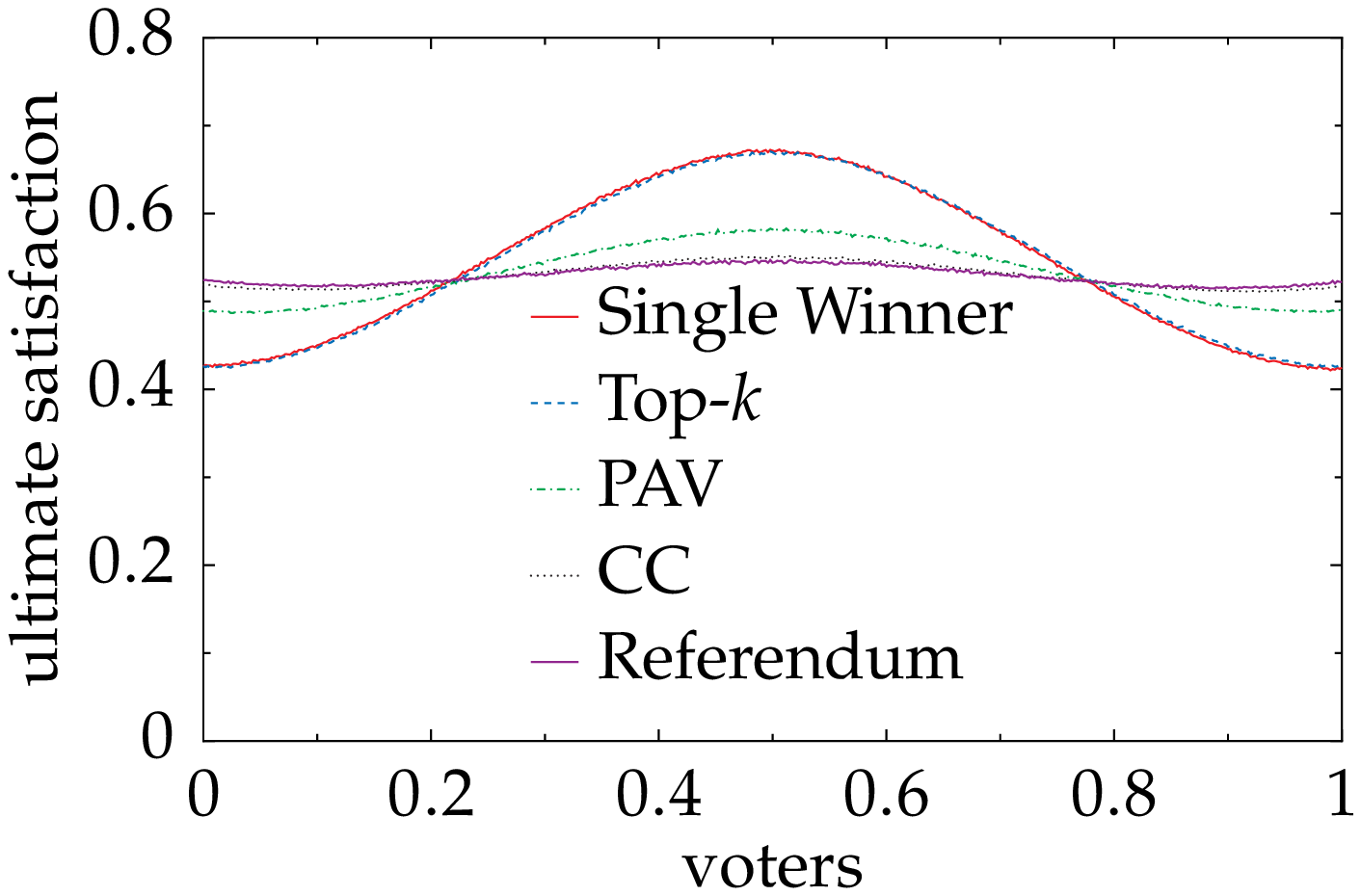}
  \caption*{(b) random dictatorship}
\end{minipage}

\caption{The results of the simulations for voters and candidates uniformly distributed on the Euclidean interval $[0, 1]$. All issues are important. Points on the $x$-axis correspond to voters.}
\label{fig:simulations}
\end{figure*}

The results of our simulations are depicted in Figure~\ref{fig:simulations}. Our conclusions are the following:
\begin{enumerate}
\item For the case of majority rule, the committee elected by PAV strictly dominates the single-winner committee. This justifies the decision of selecting a collective body rather than entrusting the whole power to a single representative. Further, this observation suggest a conjecture that an analogous variant of Condorcet's jury theorem can be formulated for the voting committee model. We recall that the main difference is that in our model there is no absolute criterion suggesting that some voter is right or wrong with respect to a given issue. Even for a high value of parameter $p$ (cf. point 2 in the description of simulations), where the issue is relatively unpopular, a voter who is close to such issue will still prefer it to be accepted.       
\item The areas below the plots for all four multiwinner committees are almost the same. Yet, the committee selected by the Chamberlin--Courant rule ensures more fairness to the voters; the two remaining multiwinner rules tend to favor centrist voters. This observation suggests the conjecture that the distribution of the members of the elected committee in a spatial model should resemble the distribution of the voters.
\item Decisions made by the majority rule are in expectation better than decisions made by the random dictatorship rule. 
\end{enumerate}

Next, we wanted to check how the intensities of voters' preferences are reflected in the voting committee model. Intuitively, a voter has stronger feelings about issues which are either very close to him or her, or which are far away on the line of preferences. A voter cares more about decisions made with respect to such issues in comparison to issues for which he or she is relatively indifferent. 
We modeled intensities of voters' preferences by introducing insignificant issues to our simulations. For an issue $i$ and voter $v$ we assumed that $i$ is insignificant for $v$ if $p \cdot |i - x| \in [\nicefrac{2}{5}, \nicefrac{3}{5}]$---intuitively, in such case $i$ is insignificant because the probabilities of $v$ willing to accept and reject the issue are relatively similar.
The results of simulations with this modification showed almost the same shapes as in Figure~\ref{fig:simulations}. These results confirm our previous conclusions. Even more, in this case the area below the plot for the Chamberlin--Courant committee was by $4\%$ larger than for the top-$K$ committee. This suggests that the Chamberlin--Courant committee is not only more fair to the voters, but it is also superior to the centrist committee with respect to the total satisfaction of voters.

\subsection{Simulations With Datasets Describing Real-Life Preferences}

As the next step we run simulations for numerous datasets describing people's preferences. To this end we used PrefLib~\cite{conf/aldt/MatteiW13}, a reference library of preference data. PrefLib contains over 300 datasets describing different scenarios where voters have (weak) preferences over candidates. We filtered out datasets with less than 15 voters and with less than 20 candidates. For each of the remaining datasets we run the following procedure:
\begin{enumerate} 
\item For each voter $v$ we introduced one issue $i_v$.
\item For an issue $i_v$ and each voter $v'$ we computed the Kendal-Tau distance~\cite{kendall1938measure} between rankings $v'$ and $v$. We scaled these distances so that the average distance between $i_v$ and all voters was equal to $\nicefrac{1}{2}$. For each candidate $c$ we defined the distance between $i_v$ and $c$ as the position of $c$ in $v$'s preference ranking. We also scaled these distances so that the average distance between $i_v$ and all candidates was equal to $\nicefrac{1}{2}$.
\item We selected a value $p$, uniformly at random from the interval $[0, 3]$. For each issue $i$ and for each voter (respectively, each candidate) $x$ we randomly selected their preference over the issue. We took $1 - p \cdot |i - x|$ as the probability of $x$ accepting $i$, where $|i - x|$ denotes the scaled distance between the voter (the candidate) and the issue, defined in the previous point.
\item Similarly as in Section~\ref{sec:prefOnLine}, we computed decisions made by different committees with respect to different issues, and for each voter $v$ we computed the average fraction of issues for each a given committee made decisions consistent with $v$'s preferences.
\item We repeated such experiment 10 times and for each voter we computed the average ultimate satisfaction.  
\end{enumerate}

\begin{figure*}[!t!h]
\begin{minipage}[h]{0.48\linewidth}
  \centering
  \includegraphics[width=\textwidth]{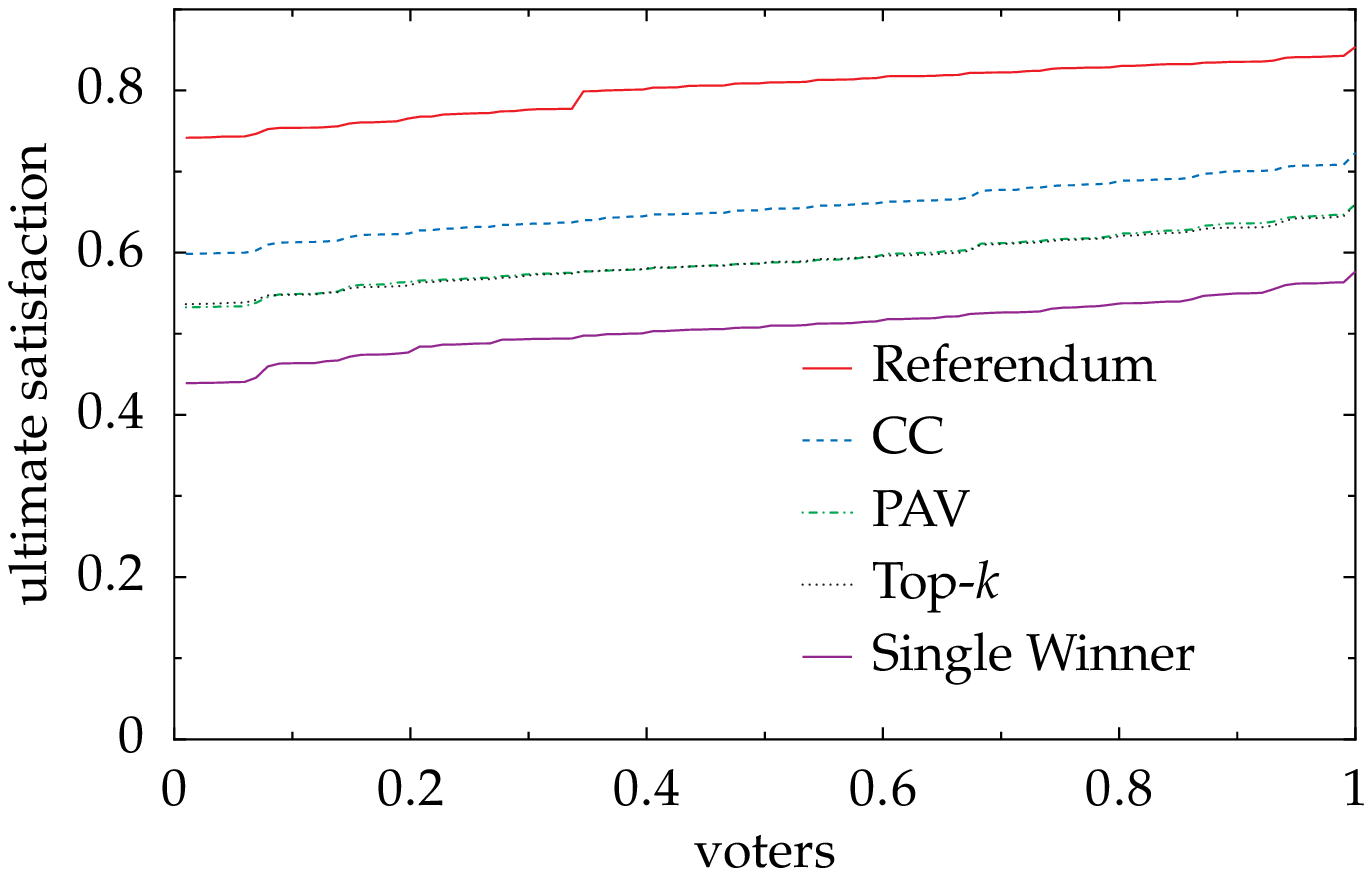}
  \caption*{(a) majority}
\end{minipage}
\hspace{0.3cm}
\begin{minipage}[h]{0.48\linewidth}
  \centering
  \includegraphics[width=\textwidth]{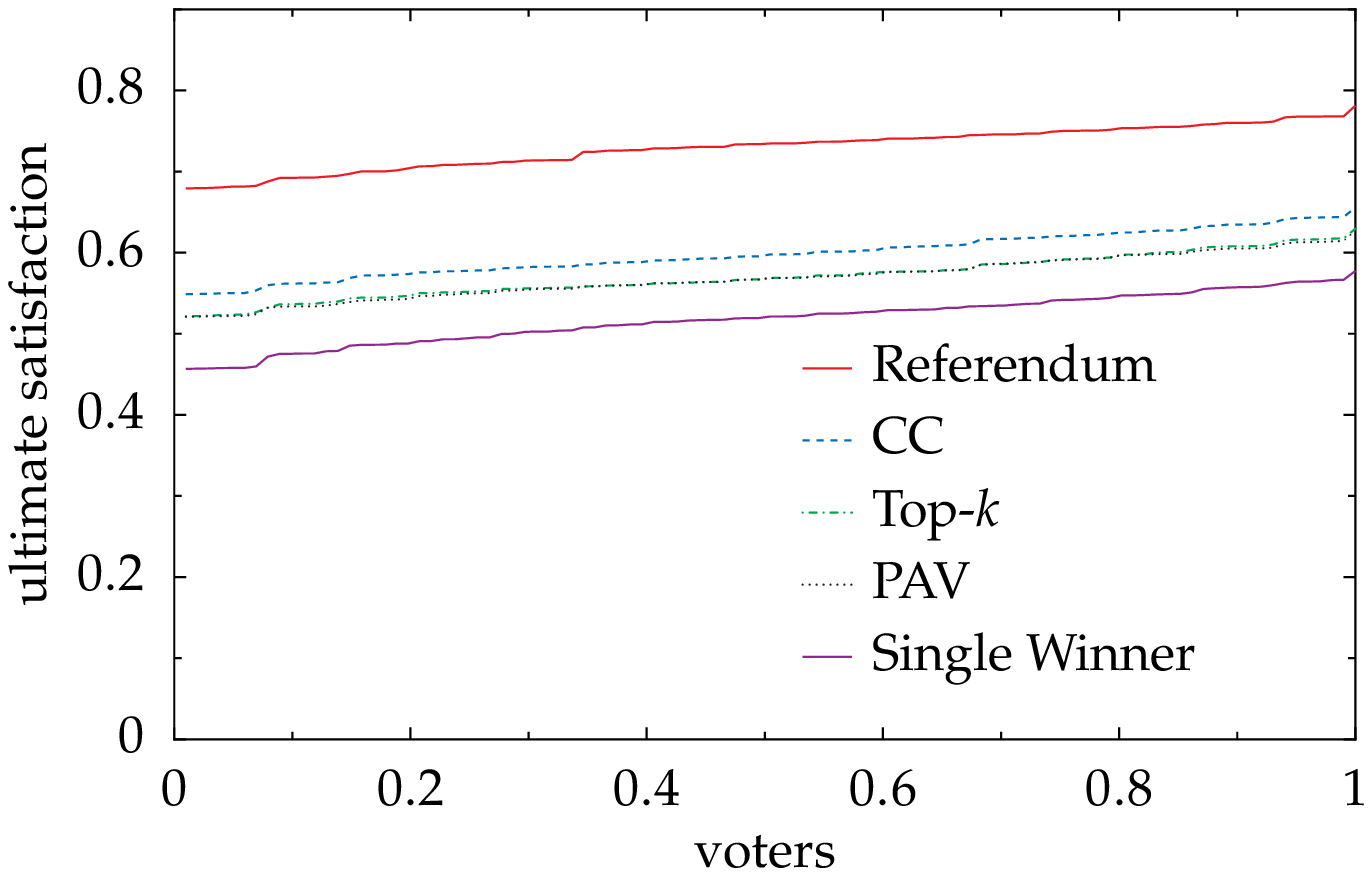}
  \caption*{(b) random dictatorship}
\end{minipage}

\caption{The results of the simulations for PrefLib~\cite{conf/aldt/MatteiW13}.}
\label{fig:real_data}
\end{figure*}

To compute the five winning committees, we took scores that voters assign to candidates, which were extracted from their ordinal preferences by using the Borda positional scoring function.

The results of these experiments are depicted in Figure~\ref{fig:real_data}. Since in this case the voters are not represented by points on the line, in Figure~\ref{fig:real_data} we sorted the voters according to the their ultimate satisfactions. For instance a point on the blue line with $x$- and $y$- coordinates equal to $0.8$ and $0.7$ respectively, denotes that 80\% of voters had ultimate satisfaction below $0.7$, and 20\% of voters above $0.7$. 

The results of these simulations confirm our conclusions from Section~\ref{sec:prefOnLine}. One significant difference is that now, we can objectively claim the Chamberlin--Courant rule to be superior, not only with respect to the fairness, but also with respect to the total ultimate satisfaction of the voters.

\subsection{Single Issue \& Deterministic Case: Theoretical Results}\label{sec:deterministicModel}

In this subsection we consider scenarios where each voter has a single important issue that is known in advance, and where voters know preferences of the candidates over the issues. Thus, when electing a committee, it is most natural to consider approval scores of the voters over candidates---a voter $i$ approves of a candidate $c$ only if the preference of $c$ over the $i$'s important issue is consistent with the preference of the voter. We recall that if a voter $i$ approves of a candidate $c$, we say that it assigns to $c$ score equal to $1$, and denote it by $u_{i, c} = 1$. Otherwise, we say that he or she assigns to $c$ score of $0$, and denote it by $u_{i, c} = 0$. These scores form an input for the multiwinner election rule $\mw$, and thus they can be viewed as the only information that rule $\mw$ can use to select a committee of $K$ representatives.

We recall that for each committee $S$, and for each voter $i$, by $\prob_{S, \sw}(i)$ we denote the probability that $S$ makes a decision consistent with $i$'s preferences, and that we consider $\prob_{S, \sw}(i)$ the ultimate satisfaction of voter $i$ from committee $S$. Since we consider the model with approval scores, it holds that:
\begin{align*}
\textstyle \prob_{S, \sw}(i) = \sw\Big(\sum_{c \in S} u_{i, c}\Big) \textrm{.}
\end{align*}

In our subsequent discussion we will show that several known multi-winner election rules can be viewed as optimal in the deterministic variant of our voting committee model. Each such a rule is optimal for different decision system $\sw$, used by the selected committee to make final decisions. Thus, our results give intuition regarding the applicability of different multi-winner election rules, depending on for what kind of decision making the committee is elected for.

\begin{definition}
Let $\sw$ be a decision rule used in the second stage of the election model. A multi-winner election rule $\mw$ is optimal for $\sw$ if for each preferences of voters it elects an optimal committee.
\end{definition}

Below, we explain which multiwinner rules are optimal for majority, random dictatorship, and unanimity decision systems.

\begin{theorem}\label{thm:majority}
Assume $K$ is odd. In the deterministic model:
\begin{inparaenum}[(i)]
\item the $\frac{K+1}{2}$-median election system is optimal for the majority rule, and
\item the top-$K$ rule is optimal for the random dictatorship rule.
\end{inparaenum}
\end{theorem}
\begin{proof}
We start by considering the case when $\sw$ is the majority rule.
We calculate $\prob_{S, \sw}(i)$, the ultimate satisfaction of a voter $i$ from a committee $S$, assuming $S$ uses majority rule to make final decisions. A committee member $c \in C$ votes according to $i$'s preferences if and only if $c$ is approved by $i$. Thus, a committee $S$ makes decisions consistent with $i$'s preferences, if it contains at least $\frac{K+1}{2}$ members approved by $i$.
The satisfaction $\prob_{S, \sw}(i)$ of $i$ from $S$ is equal to $1$ if $S$ contains at least $\frac{K+1}{2}$ members approved by $i$, otherwise it is equal to $0$.

The same formula defines satisfaction of $i$ from $S$ in the $\frac{K+1}{2}$-median election system. Finally, we note that in the $\frac{K+1}{2}$-median election system, the committee that maximizes voters' total satisfaction is selected.

Let us now move to the case when $\sw$ is the random dictatorship rule.
For a committee $S$, let $\mathrm{apprv}_i(S)$ denote the number of candidates from $S$ that are approved of by $i$. For each committee $S$, and each member $c \in S$ the probability that during the uniformly random dictatorship ballot regarding the issue, the final decision will be made by the committee member $c$ is equal to $\nicefrac{1}{K}$. The probability that $S$ will make decision according to $i$'s preferences is, thus, equal to:
\begin{align*}
\prob_{S, \sw}(i) &= \sum_{c \in S} \frac{1}{K} \cdot \indicator_{\text{$i$ approves of $c$}} = \frac{\mathrm{apprv}_i(S)}{K} \textrm{.}
\end{align*}
Consequently, a committee $S$ is optimal if it maximizes the value of the formula $\sum_i\mathrm{apprv}_i(S)$. Exactly such committee is elected by the top-$K$ rule.
\end{proof}

Theorem~\ref{thm:majority} is quite powerful in a sense that it claims the optimality of the $\frac{K+1}{2}$-median and top-$K$ multi-winner rules, for the corresponding decision rules, irrespectively of the preferences of the voters over the issues. Unfortunately, this is not always the case, which is illustrated in the following example.

\begin{example}\label{ex:unanimityNoOptimal}
Consider an election with two voters, three candidates, and where the goal is to select a committee with two candidates ($K = 2$). Assume the elected committee will use the unanimity rule to make decisions on issues. Further, assume the deterministic model with two issues, $D_1$ and $D_2$; issue $D_1$ is important only for voter $i_1$ and issue $D_2$ only for voter $i_2$. Consider preferences of the voters and of the candidates over issues given in the left table below.

\begin{minipage}{0.48\textwidth}
\centering
        \begin{tabular}{c | c c c c c} 
        \hline
        & $i_1$ & $i_2$ & $c_1$ & $c_2$ & $c_3$ \\
        \hline\hline
        $D_1$ & \emph{\textbf{R}} & & \emph{R} & \emph{R} & \emph{A} \\ 
        \hline
        $D_2$ & & \emph{\textbf{R}} & \emph{A} & \emph{A} & \emph{R} \\
        \hline
        \end{tabular}
\end{minipage}
\begin{minipage}{0.48\textwidth}
\centering
        \begin{tabular}{c | c c c c c} 
        \hline
        & $i_1$ & $i_2$ & $c_1$ & $c_2$ & $c_3$ \\
        \hline\hline
        $D_1$ & \emph{\textbf{A}} & & \emph{A} & \emph{A} & \emph{R} \\ 
        \hline
        $D_2$ & & \emph{\textbf{A}} & \emph{R} & \emph{R} & \emph{A} \\
        \hline
        \end{tabular}
\end{minipage}
\medskip

\noindent
Observe that voter $i_1$ approves of candidates $c_1$ and $c_2$ and that voter $i_2$ approves of candidate $c_3$. In this case the optimal committee would be $\{c_1, c_3\}$ or $\{c_2, c_3\}$. Both committees will reject both issues, and so they will satisfy both voters. Now, consider the preferences from the right table. Similarly as before, in this case voter $i_1$ approves of candidates $c_1$ and $c_2$ and voter $i_2$ approves of candidate $c_3$. However, in this case the optimal committee is $\{c_1, c_2\}$. This committee will accept issue $D_1$, hence satisfy one voter. Any other committee will reject both issues satisfying no voter.
\end{example}

Example~\ref{ex:unanimityNoOptimal} suggests that for the unanimity rule there exists no optimal multiwinner election method. This is an artifact of the fact that unanimity is not symmetric with respect to decisions $\acc$ and $\rej$. For the unanimity rule we can obtain a much weaker claim, yet our claim will also have an interesting interpretation. Before we proceed further, we introduce two new definitions that describe two extreme classes of voters' preferences over the issues. We say that voters are \emph{rejection-oriented} if for each voter $i$, $d_i = \rej$, meaning that each voter gets satisfaction only from rejecting issues. Analogously, we say that the voters are \emph{acceptance-oriented} if for each voter $i$, $d_i = \acc$.

\begin{proposition}\label{thm:chambCourVeto}
For the unanimity system in the deterministic model:
\begin{inparaenum}[(i)]
\item the Chamberlin--Courant system with approval votes is optimal for rejection-oriented voters, and
\item the $K$-median system with approval votes is optimal for acceptance-oriented voters.
\end{inparaenum}
\end{proposition}
\begin{proof}
We provide the proof for rejection-oriented voters. The proof for acceptance-oriented voters is analogous.
Our reasoning is very similar to the proof of Theorem~\ref{thm:majority}. The ultimate satisfaction of a voter $i$ from a committee $S$, $\prob_{S, \sw}(i)$, is equal to 1 if $S$ contains at least one candidate approved by $i$, and is equal to $0$ otherwise. This is equivalent to the definition of voters' satisfaction in the Chamberlin--Courant election rule, which completes the proof.
\end{proof}


Proposition~\ref{thm:chambCourVeto} suggest that Chamberlin--Courant rule and $K$-median rule are suitable for electing committees that have veto rights. One should choose one of them depending on the voters' satisfaction model. For instance, if passing a wrong decision has much more severe consequences than rejecting a good one (which is captured by modeling the voters as rejection-oriented), a committee should be selected with Chamberlin--Courant rule, and it should use unanimity rule to make final decisions. This is very often the case when there exists a well established status quo which should be changed only in indisputable cases. 
On the other hand, if passing a wrong decision has relatively low cost compared to rejecting a good one, the committee should be selected with $K$-median rule.

\section{Independent Voters: Theoretical Analysis}

In this section we consider voters whose preferences over issues can be viewed as independent random variables. We will show that if the voters are independent, there exist optimal rules, also in the nondeterministic case. In what it follows, we relax two assumptions from Section~\ref{sec:deterministicModel}. First, we assume that each voter can consider multiple issues as important. Second, we assume that the voters do not know how the candidates are going to vote. Nevertheless, they are provided with some form of intuition which allows them to recognize that some candidates would better represent them in the elected committee than others. To capture this intuition, for each voter $i$ and each candidate $c$ we define the \emph{probability of representation of $i$ by $c$}, denoted by $p_{i, c}$. This is the probability that $c$ will vote according to $i$'s preferences on an $i$'s important issue. We write $q_{i, c} = (1 - p_{i, c})$.
Clearly, value $p_{i, c}$ measures how well $i$ feels represented by $c$, thus we say that voter $i$ assigns score $u_{i, c} = p_{i, c}$ to candidate $c$. These scores are given as an input to the multiwinner election rule $\mw$.


Further, for each voter $i$, each important for $i$ issue $D_j \in \important(i)$, each committee $S$, and each committee vote $v \in D_j^K$, let $\prob_{S}(v)$ denote the probability that members of $S$ cast vote $v$:
\begin{align*}
\prob_{S}(v) = \prod_{c \in S} \left(\indicator_{v[j] = d_i^j}p_{i, c} + \indicator_{v[j] \neq d_i^j}q_{i, c}\right) \textrm{.}
\end{align*}
Recall that $\sw: \naturals \to \reals$ is a randomized decision rule used by the committee to make decisions over issues. By $\prob_{\sw}(i | v)$ we denote the probability that the rule $\sw$, given vote $v$, makes decision $d_j^i$, i.e., decision consistent with $i$'s preferences:
\begin{align*}
\prob_{\sw}(i | v) &= \indicator_{d_i^j = \acc} \sw\Big(\sum_{d \in v} \indicator_{d = \acc}\Big) + \indicator_{d_i^j = \rej} \Big(1 - \sw\Big(\sum_{d \in v} \indicator_{d = \acc}\Big)\Big) \\
                   &= \indicator_{d_i^j = \acc} \sw\Big(\sum_{d \in v} \indicator_{d = \acc}\Big) + \indicator_{d_i^j = \rej} \Big(\sw\Big(\sum_{d \in v} \indicator_{d = \rej}\Big)\Big) \\
                   &= \sw\Big(\sum_{d \in v} \indicator_{d = d_{i}^j}\Big) \textrm{.}
\end{align*}
We recall that for each committee $S$, each voter $i$, and each important issue $D_j \in \important(i)$,
by $\prob_{S, \sw}(i)$ we denote the probability that $S$ makes decision consistent with $i$'s preferences. Thus:
\begin{align*}
\prob_{S, \sw}(i) = \sum_{v \in D_j^K} \prob_{S}(v)\prob_{\sw}(i | v)\textrm{.}
\end{align*}

\noindent $\prob_{S, \sw}(i)$ can be also viewed as the expected fraction of issues important for $i$, for which the committee $C$ would make decisions consistent with $i$'s preferences. Similarly as in the deterministic model, we will call $\prob_{S, \sw}(i)$ the ultimate satisfaction of a voter $i$ from the committee $S$.

\begin{example}
Consider the illustrative example from Section~\ref{sec:correlatedVoters} where one third of voters are left-wing, one third are right-wing, and one third are centrists. Recall that left-wing, right-wing, and centrists candidates can be represented by points 0, $\nicefrac{1}{2}$, and 1 on the line, respectively. For this example the probability of representation $p_{i, c}$ is equal to $p_{i, c} = 1 - d(i, c)$, where $d(i, c)$ is the distance between the points corresponding to voter $i$ and candidate $c$.
\end{example}

The \emph{approval model} is a special case of the probabilistic model where there are only two ``allowed'' values of scores. Thus, in the approval model we assume that there exists two values $p, q \in [0, 1]$ with $q = p-1$, such that for each voter $i$ and each candidate $c$ we have $p_{i, c}, q_{i, c} \in \{p, q\}$. This model captures scenarios when each voter can view each candidate as belonging to one of the two extreme categories: each candidate can be either ``good'' or ``bad'' from $i$'s point of view. This model is appealing because in many real-life scenarios the voters express their preferences by assigning candidates to one of the two groups: either by approving or by disapproving them.

In contrast to the previous section, here we assume that the scores that the voters assign to candidates are scaled: we assume that a voter $i$ assigns score $1$ to candidate $c$ if $p_{i, c} = p$ and score $0$ if $p_{i, c} = q$. These scores will be given as an input to the multiwinner election rule $\mw$. The fact that we use rescaled scores will not affect our further results, but will enable us to use a more readable notation.

\subsection{Optimality of Known Decision Rules under Approval Voting}

We start our analysis with the approval model, i.e., from the case when the voters have two types of scores only. 
We start by observing that the characterization of the optimal rule for the random dictatorship method from Theorem~\ref{thm:majority} can be generalized to the approval model.

\begin{proposition}\label{thm:kapproval}
The top-$K$ rule is optimal for the random dictatorship rule in the approval model.
\end{proposition}
\begin{proof}
Similarly as in the proof of Theorem~\ref{thm:majority}, for a committee $S$ we define $\mathrm{apprv}_i(S)$ as the number of candidates from $S$ that are approved by $i$. Let $p$ and $q$ denote the probabilities of representation of $i$ by candidates approved of and disapproved of by $i$, respectively. Naturally, $p > q$. For each committee $S$, and each member $c \in S$ the probability that during the uniformly random dictatorship ballot regarding an issue, the final decision will be made by the committee member $c$, is equal to $\nicefrac{1}{K}$. The probability that $S$ will make decision according to $i$'s preferences is equal to:
\begin{align*}
\prob_{S, \sw}(i) &= \sum_{c \in S} \frac{1}{K} \cdot (p\indicator_{\text{$i$ approves of $c$}} + q\indicator_{\text{$i$ disapproves of $c$}}) \\
                            &= \mathrm{apprv}_i(S) \cdot \frac{p}{K} + (K - \mathrm{apprv}_i(S)) \cdot \frac{q}{K} \\
                            &= \mathrm{apprv}_i(S) \cdot \frac{p - q}{K} + q\textrm{.}
\end{align*}
Consequently, a committee $S$ is optimal if it maximizes $\sum_i\mathrm{apprv}_i(S)$. Exactly such committee is elected by the top-$K$ rule.
\end{proof}

Below, we show a more general result that characterizes a class of optimal election systems.

\begin{theorem}\label{thm:expressivness}
For each decision rule $\sw$, there exists a $K$-element vector $\alpha$, such that $\alpha$-rule is optimal for $\sw$ in the approval model. 
\end{theorem}
\begin{proof}
Let $S_{\ell, i}$ denote a committee that has exactly $\ell$ members approved by $i$, and let $\prob(i, s, \ell)$ be the probability that exactly $s$ members of $S_{\ell, i}$ will vote accordingly to $i$'s preferences. We have:
\begin{align*}
\prob(i, s, \ell) &= \sum_{x=1}^K \indicator_{x\leq \ell} \indicator_{x \leq s}\indicator_{s-x \leq K-\ell} \cdot {\ell \choose x} p^xq^{\ell - x} {K-\ell \choose s - x}q^{s-x}p^{K-\ell-s+x} \textrm{.}
\end{align*}
We can see that $\prob(i, s, \ell)$ does not depend on $i$. Further, let $\prob_{\sw}(i | s)$ be the probability that the rule $\sw$ makes decision consistent with $i$'s preferences on issue $D_j$, assuming $s$ members of $S_{\ell, i}$ vote accordingly to $i$'s preferences.
\begin{align*}
\prob_{\sw}(i | s) = \sw(s)\indicator_{d_i^j = \acc} + (1 - \sw(K-s)) \indicator_{d_i^j = \rej} \textrm{.}
\end{align*}
Since either $d_i^j = \acc$ or $d_i^j = \rej$ and since $\sw(s) = (1 - \sw(K-s))$, we get that $\prob_{\sw}(i | s) = \sw(s)$, and thus $\prob_{\sw}(i | s)$ does not depend on $i$.
Consequently,  we can calculate the ultimate satisfaction $\prob_{S_{\ell, i}, \sw}(i)$ of a voter $i$ from a committee $S_{\ell, i}$, so that this value does not depend on $i$:
\begin{align*}
\prob_{S_{\ell, i}, \sw}(i) = \sum_{s = 1}^{K}\prob(i, s, \ell) \prob_{\sw}(i | s) \textrm{.}
\end{align*}
Because $\prob_{S_{\ell, i}, \sw}(i)$ does not depend on $i$, we will denote it as $\prob_{S_{\ell}, \sw}$.
Naturally, since $p \geq q$, we have $\prob_{S_{\ell+1}, \sw} \geq \prob_{S_\ell, \sw}$, for each $\ell$. Now, we can see that the following vector:
\begin{align*}
\alpha = \Big\langle &\prob_{S_{1}, \sw}, \;(\prob_{S_{2}, \;\sw} - \prob_{S_{1}, \sw}), \;(\prob_{S_{3}, \sw} - \prob_{S_{2}, \sw}), \\
                 & \dots, (\prob_{S_{K}, \sw} - \prob_{S_{K-1}, \sw})\rangle
\end{align*}
satisfies the requirement from the thesis. Indeed, in the $\alpha$-rule the satisfaction of a voter from a committee with $\ell$ approved members is the sum of first $\ell$ coefficients of $\alpha$, which is $\prob_{C_\ell, \sw}$. This completes the proof.
\end{proof}

The above theorem can be viewed as an evidence of the expressiveness and power of $\owa$ election rules. Unfortunately, $\owa$ election rules are not sufficiently expressive to describe the non-approval model, which we address in the next subsection.

\subsection{Optimality of Known Decision Rules under Probabilistic Model} 

Let us move to the probabilistic model in its full generality.

To get characterization similar to the one given in Theorem~\ref{thm:expressivness}, but for arbitrary scores, we would need to consider a more general class of election rules, that is rules in which the satisfaction of a single voter $i$ from a committee $C$ is expressed as a linear combination of products of $i$' scores assigned to individuals, i.e., as a linear combination of the values from the set $\big\{\prod_{c \in C'}u_{i, c}\colon C' \subseteq C\big\}$ (in contrast to a linear combination of scores only). Such rules that consider inseparable committees were considered e.g., by Ratliff~\cite{ratliff2006selecting,ratliff2003some}. We formalize this observation in the following theorem.

\begin{theorem}\label{thm:generalOptimalRules}
Let $\sw: \naturals \to \reals$ be the rule that the elected committee $S$ uses to make final decisions. Let us consider the function $v: \pow_K(C) \times \Pi^{n}(C) \to \reals$ defined in the following way:
\begin{align}\label{eq:score}
v(S, \langle u_i\rangle_{i \in N}) = \sum_{i \in N} \sum_{j = 1}^{|S|} \sw(j) \sum_{S_{\ap} \subseteq \pow_j(S)} \Bigg( \prod_{c \in S_{\ap}} u_{i, c} \prod_{c \notin S_{\ap}}\big(1 -  u_{i, c} \big)\Bigg)\textrm{,}
\end{align}
The rule that for each score profile $\langle u_i\rangle_{i \in N}$ selects the committee $C$ that maximizes $v(C, \langle u_i\rangle_{i \in N})$ is optimal for $\sw$.
\end{theorem}
\begin{proof}
The thesis follows from the fact that $v(C, \langle u_i\rangle_{i \in N})$ is the sum of ultimate expected satisfactions over all voters. Indeed, for each possible value of $j$ representing the number of committee members voting according to $i$'s preferences, the expression:
\begin{align*}
\sum_{S_{\ap} \subseteq \pow_j(S)} \Bigg( \prod_{c \in S_{\ap}} u_{i, c} \prod_{c \notin S_{\ap}}\big(1 -  u_{i, c} \big)\Bigg)
\end{align*}
gives the probability that exactly $j$ committee members will vote according to $i$'s preferences. In the above formula $S_{\ap}$ represents the set of $j$ committee members who vote according to $i$'s preferences.
\end{proof}

Nevertheless, for the case of the uniformly random dictatorship rule used to make final decisions, we can get a result similar to the one given in Proposition~\ref{thm:kapproval} even for the case of arbitrary scores.

\begin{proposition}\label{thm:kTopGeneralUtilities}
Top-$K$ rule is optimal for the random dictatorship rule in the probabilistic model.
\end{proposition}
\begin{proof}
The reasoning is similar to the proof of Proposition~\ref{thm:kapproval}.
The probability that a committee $S$ will vote according to $i$'s preferences is equal to:
\begin{align*}
\prob_{S, \sw}(i) = \frac{1}{K} \sum_{c \in S} p_{i, c} \textrm{.}
\end{align*}
The rule that maximizes the total score of selected candidates, also maximizes their total expected ultimate satisfaction. This completes the proof.
\end{proof}

Theorem~\ref{thm:generalOptimalRules} gives us a constructive way of defining an optimal rule in the most general case. Unfortunately, such rule is often hard to compute. We briefly discuss computational issues in Appendix~\ref{sec:computationalAspects}.

\subsection{Optimality of Full Multiwinner Rules}\label{sec:indirect_rules}

In the previous subsection we studied the optimality of multi-winner election rules, given information on what decision system the committee will use to make final decisions.
In this section we show that in our probabilistic model we can compare qualities of the pairs of multi-winner and decision systems. Thus, our results not only suggest which multi-winner election rule is suitable for electing a committee, but also indicates which decision rule should be used by the committee to make final decisions.
We start with introducing the definition of the full multiwinner rule. This definition introduces a novel concept to the literature on social choice: it allows the voters to elect not only committees, but also decision rules that these committees use to make final decisions.

\begin{definition}
Let $\mathcal{SW}$ be the set of all symmetric randomized decision rules.
A \emph{full multiwinner rule} is a function $\mathcal{F}: \Pi \to \mathcal{P}(A) \times \mathcal{SW}$ that for each score profile $\pi \in \Pi$ returns a pair $\mathcal{F}(\pi) = (S, \sw)$, where $S \in \mathcal{P}(C)$ is the elected committee, and $\sw \in \mathcal{SW}$ is the decision rule that $S$ will use when making final decisions. We use the notation $\mathcal{F}(\pi)[1] = S$ and $\mathcal{F}(\pi)[2] = \sw$.
\end{definition}

Below we define a partial order on full multiwinner rules. This definition provides a way to compare full multiwinner rules, in particular it allows us to
extend the concept of optimality to full multiwinner rules. 

\begin{definition}\label{def:optimality}
A full multiwinner rule $\mathcal{F}$ \emph{weakly dominates} a full multiwinner rule $\mathcal{G}$, if for each score profile $\pi \in \Pi(N)$, $\mathcal{F}$ returns a committee and a decision rule that gives the total expected ultimate satisfaction of the voters at least as high as the one given by a committee and a decision rule returned by $\mathcal{G}$:
\begin{align}\label{eq:weakDominant}
\sum_i\prob_{\mathcal{F}(\pi)[1], \mathcal{F}(\pi)[2]}(i) \geq \sum_i\prob_{\mathcal{G}(\pi)[1], \mathcal{G}(\pi)[2]}(i) \textrm{.}
\end{align}
$\mathcal{F}$ \emph{strongly dominates} $\mathcal{G}$ for $P$ if it weakly dominates $\mathcal{G}$ and if there exists profile $\pi \in \Pi$ for which Inequality~\eqref{eq:weakDominant} is strict.
A full multiwinner rule $\mathcal{F}$ is \emph{optimal} if it weakly dominates every other full multiwinner rule.
\end{definition}

In the previous sections we assumed that the decision rule used by the committee is fixed, and that we use voters' preferences only to determine the winning committee of $K$ candidates. This traditional definition can be mapped to the new model.
 
\begin{definition}
For a multi-winner rule $\mw$ and a decision rule $\sw$, by $\mw$ followed by $\sw$ we call the full multiwinner rule that selects committee using $\mw$ and, independently of voters' preferences, always uses $\sw$ to make final decisions.
\end{definition}

Now, we show how to apply Definition~\ref{def:optimality} in perhaps the simplest variant, that is in the deterministic model. First, we define a new election rule $\textsc{Comb}$ as a combination of $\frac{K+1}{2}$-median rule followed by majority rule with the $K$-approval rule followed by the random dictatorship rule.

\begin{definition}\label{def:combination}
The rule $\textsc{Comb}$ is defined as follows. Let $S$ and $S'$ be committees elected by $K$-approval rule and by $\frac{K+1}{2}$-median rule, respectively. Let $\mathrm{apprv}$ be the total approval score of $S$ and let $\mathrm{owa}$ be the total OWA score of $S'$. If $\frac{\mathrm{apprv}}{K} > \mathrm{owa}$, then $\textsc{Comb}$ returns the pair $(C$, $\mathrm{random}$ $\mathrm{dictatorship})$. Otherwise, $\textsc{Comb}$ returns the pair $(C'$, $\mathrm{majority})$. 
\end{definition}

It is remarkable that with this simple idea we obtained the new rule, $\textsc{Comb}$, that strongly dominates both rules that it is derived from.

\begin{proposition}
In the deterministic model, $\textsc{Comb}$ strongly dominates $\frac{K+1}{2}$-median followed by majority, and $K$-approval followed by random dictatorship.
\end{proposition}
\begin{proof}
Let $\mathrm{apprv}$ and $\mathrm{owa}$ be defined as in Definition~\ref{def:combination}. Repeating the analysis from the proofs of Theorems~\ref{thm:majority} and Proposition~\ref{thm:kapproval}, we get that the total ultimate satisfaction of voters under $\frac{K+1}{2}$-median followed by majority is equal to $\mathrm{owa}$, and that the total ultimate satisfaction of voters under $K$-approval followed by random dictatorship is equal to $\frac{\mathrm{apprv}}{K}$. The total ultimate satisfaction of voters and under $\textsc{Comb}$ is equal to $\max(\frac{\mathrm{apprv}}{K}, \mathrm{owa})$. It is easy to see that there exist profiles where $\frac{\mathrm{apprv}}{K}$ is strictly greater than $\mathrm{owa}$ and vice versa.
This completes the proof.
\end{proof}

The natural question is whether we can find an optimal rule in our general probabilistic model. Interestingly, for the deterministic model the answer is positive. The question whether this result can be extended to the general probabilistic model is still open.

\begin{theorem}
In the deterministic model there exists an optimal full multiwinner rule.
\end{theorem}
\begin{proof}
The proof is constructive. We recall that a sin\-gle-win\-ner rule $\sw$ can be described by $K$ values $\sw(1), \dots, \sw(K)$. We use notation from the proof of Theorem~\ref{thm:expressivness}; let $S_{\ell, i}$ denote a committee that has exactly $\ell$ members approved by $i$, and let $\prob(i, s, \ell)$ denote the probability that exactly $s$ members of $S_{\ell, i}$ will vote accordingly to $i$'s preferences.

In the deterministic model $\prob(i, s, \ell) = \indicator_{s=\ell}$. Consequently, for a committee $S_{\ell, i}$ with $\ell$ approved members, we have $\prob_{S_{\ell, i}, \sw} = \sw(\ell)$. Let $\indicator_{\ell, i}$ denote a function such that $\indicator_{\ell, i}(S) = 1$ if $S$ contains $\ell$ elements approved by $i$, and $\indicator_{\ell, i}(S) = 0$ otherwise. For a given committee $S$ we can find an optimal decision rule $\sw$ by solving the following linear program:
{\upshape
  \begin{align*}
    &\text{maximize }    \sum_{i \in N}\prob_{S, \sw}(i) = \sum_{i \in N}\sw(\ell) \indicator_{\text{$i$ approves $\ell$ elements of $S$}} \\
    &\text{subject to: }\\  
    & \text{(a)}: \sw(\ell+1) \geq \sw(\ell),   \hspace{1.3cm}   1 \leq \ell \leq K-1 \\
    & \text{(b)}: 0 \leq \sw(\ell) \leq 1,                \hspace{2.6cm}   1 \leq \ell \leq K \\
    & \text{(c)}: \sw(\ell) = 1 - \sw(K - \ell),   \hspace{1.2cm}   1 \leq \ell \leq K \\
    & \text{(c)}: \sum_{\ell=1}^K\sw(\ell) = 1 \text{.}
  \end{align*}
}
The optimal full multiwinner rule tries all committees and selects such that gives the best solution to the integer program. From the solution of the integer program we can extract values $\sw(1), \dots, \sw(K)$ that describe the optimal decision rule that should be used to make final decisions.
\end{proof}

\section{Discusion \& Conclusion}

We defined a new model, called the voting committee model, which explores scenarios where a group of representatives is elected to make decisions on behalf of the voters. This model links utilities of the voters from the elected committee to their utilities from the committee's decisions regarding various matters. Intuitively, a satisfaction of voter $i$ from a committee $S$ is proportional to the probability that for the issues important for $i$, $S$'s decisions are consistent with $i$'s preferences.

Our results give positive support for employing representation-focused multiwinner election rules---indeed under several assumptions, the decisions made by such committees reflect the preferences of the population best. Most importantly, however, our model introduces a new framework for normative comparison of multiwinner voting rules. Indeed, the analysis presented in this paper can be repeated for more specific scenarios, where the relation between voters' preferences over candidates and respective probabilities of representation exhibits some specific structure.

We introduced the notion of a full multiwinner rule, in which voters elect a decisive committee along with the decision rule to be used. Under some assumptions, we showed that there exists an optimal full multiwinner rule and we described the way how it can be constructed. We believe that this is an interesting concept that deserves further attention.

Among many natural open questions, there is one we consider particularly appealing: is it possible to provide theoretical analysis suggesting which committee is the best for the case of spatial model, such as the one considered in Section~\ref{sec:correlatedVoters}?

\smallskip\noindent\textbf{Acknowledgements.}  The author thanks Jean Fran\c{c}ois Laslier, Marcin Dziubi\'nski and Piotr Faliszewski for their comments and for the fruitful discussion.

The author was supported by the European Research Council grant ERC-StG 639945.

\bibliographystyle{abbrv}
\bibliography{models}

\appendix
\section{Computational Aspects of Comparing Committees}\label{sec:computationalAspects}

Unfortunately, finding optimal committees is often computationally hard. For example, the problem of finding winners under Chamberlin--Courant rule is known to be $\np$-hard~\cite{complexityProportionalRepr} and hard from the point of view of parameterized complexity theory~\cite{fullyProportionalRepr}. On the positive side, this problem can be effectively solved in practice by using high-quality polynomial-time~\cite{sko-fal-sli:j:multiwinner} and fixed-parameter-tractable-time~\cite{conf/aaai/SkowronF15} approximation algorithms. Regrettably, not all variants of the problem of finding optimal committees can be well-approximated. For instance, Skowron~el~al.~\cite{owaWinner} showed that there exists no polynomial-time constant-approximation algorithm for the problem of finding winners under the $\frac{K+1}{2}$-median rule.

Even though the problem of finding optimal committees is computationally intractable, we can use our findings to compare committees, even in the case of arbitrary utilities. Indeed, Proposition~\ref{thm:computingScore} below shows that deciding which one of the given two committees is better according to our model is solvable in polynomial time. This observation can be used to make decisions about which from the several shortlisted committees should be selected in a particular instance of elections, or to derive more general conclusions about applicability of different multi-winner rules by comparing their qualities on real data describing voters' preferences~\cite{mat-wal:c:preflib}.

\begin{proposition}\label{thm:computingScore}
For a decision rule  $\sw: \naturals \to \reals$, a committee $S$, and a utility profile $\langle u_i\rangle_{i \in N}$, the value of the function $v$ defined in Equation~\eqref{eq:score} can be computed in polynomial time.
\end{proposition}
\begin{proof}
For each voter $i \in N$ and each committee $S$, the expected ultimate satisfaction of $i$ from $S$ can be computed by dynamic programming. Let us sort the members of $S$ in some arbitrary order, so that $S = \{c_1, c_2, \ldots, c_K\}$. For each $j$, $0 \leq j \leq K$ and for each $\ell$,  $0 \leq \ell \leq K$, we define the value $A[j][\ell]$ as the probability that exactly $\ell$ members from the set $\{c_1, \ldots, c_j\}$ will vote according to $i$'s preferences. Naturally, $A[0][0] = 1$, and for $\ell > 0$, $A[0][\ell] = 0$. To compute the remaining values of the table $A$ we can use the following relation:
\begin{align*}
A[j][\ell] = u_{i, c_{j}} \cdot A[j-1][\ell-1] + (1 - u_{i, c_{j}}) \cdot A[j-1][\ell] \textrm{.} 
\end{align*}
The value of the expected ultimate satisfaction of $i$ can be obtained by computing a linear combination of the values $A[K][\ell]$, i.e., by computing $\sum_{\ell = 1}^{|A|} \sw(\ell) \cdot A[K][\ell]$. This completes the proof.
\end{proof}

\end{document}